\documentclass{article}
\usepackage[a4paper,top=30mm,bottom=35mm]{geometry}
\usepackage[english]{babel}
\usepackage{amsmath,amssymb}
\usepackage{amsthm,thmtools}
\usepackage{lmodern,microtype}
\usepackage[normalem]{ulem}
\usepackage{enumitem}
\usepackage{tikz}
\usetikzlibrary{cd}
\usepackage[colorlinks=true,linkcolor=blue!40!black,citecolor=blue!40!black,urlcolor=blue!40!black]{hyperref}
\usepackage[capitalise,noabbrev,nameinlink]{cleveref} 

\hypersetup{hypertexnames=false} 

\newcommand{\hair}{\ifmmode\mskip1mu\else\kern0.08em\fi}  

\makeatletter
\newcommand{\authAaddr}[1]{\gdef\authA@address{\par\textsc{#1}}}
\newcommand{\authAemail}[1]{\gdef\authA@email{\par\textit{E-mail address:} \texttt{\href{mailto:#1}{#1}.}}}
\newcommand{\authBaddr}[1]{\gdef\authB@address{\par\textsc{#1}}}
\newcommand{\authBemail}[1]{\gdef\authB@email{\par\textit{E-mail address:} \texttt{\href{mailto:#1}{#1}.}}}
\newcommand{\authCaddr}[1]{\gdef\authC@address{\par\textsc{#1}}}
\newcommand{\authCemail}[1]{\gdef\authC@email{\par\textit{E-mail address:} \texttt{\href{mailto:#1}{#1}.}}}
\AtEndDocument{
	\par\bigskip\authA@address\authA@email
	\par\bigskip\authB@address\authB@email
	\par\bigskip\authC@address\authC@email
}
\makeatother

\author{Josse van Dobben de Bruyn \and Remy van Dobben de Bruyn \and Peter Zeman}
\title{Mermin--Peres magic rectangles modulo odd primes}
\date{5 October 2026}

\authAaddr{Department of Applied Mathematics, Faculty of Mathematics and Physics, Charles University, Czech Republic}
\authAemail{josse.van-dobben-de-bruyn@matfyz.cuni.cz}

\authBaddr{Max Planck Institute for Mathematics, Bonn, Germany}
\authBemail{dobbendebruyn@mpim-bonn.mpg.de}

\authCaddr{Department of Algebra, Faculty of Mathematics and Physics, Charles University, Czech Republic}
\authCemail{peter.zeman@matfyz.cuni.cz}

\declaretheorem[style=definition,numberwithin=section]{definition}
\declaretheorem[style=definition,numberlike=definition]{construction}
\declaretheorem[style=definition,numberlike=definition]{question}

\declaretheorem[style=plain,numberlike=definition]{proposition}
\declaretheorem[style=plain,numberlike=definition]{lemma}
\declaretheorem[style=plain,numberlike=definition]{theorem}
\declaretheorem[style=plain,numberlike=definition]{corollary}
\declaretheorem[style=plain,numberlike=definition]{conjecture}

\declaretheorem[style=remark,numberlike=definition]{remark}

\newcommand{\Z}{\mathbb{Z}}
\newcommand{\C}{\mathbb{C}}
\newcommand{\F}{\mathbb{F}}
\newcommand{\OO}{\mathcal O}

\DeclareMathOperator{\id}{id}
\DeclareMathOperator{\spn}{span}
\DeclareMathOperator{\rank}{rank}
\DeclareMathOperator{\ev}{ev}

\begin{document}

\maketitle

\begin{abstract}
	The Mermin--Peres magic square provides a simple example of a system of linear equations over $\Z/2\Z$ which has no classical solutions but does have a finite-dimensional operator solution.
	For a long time, it was not known how to construct similar examples over $\Z/d\Z$ with $d$ odd.
	In this paper, we construct, for every integer $d \geq 2$, a linear system over $\Z/d\Z$ that has a finite-dimensional operator solution but no classical solution.
	For an odd prime $p$, our operators act on two $p$-dimensional qudits and generate a finite $p$-group obtained by adjoining diagonal polynomial phase operators to the generalized Pauli group, thereby forming a natural analogue of the Mermin--Peres magic square.
	Classical inconsistency follows from an elementary linearity argument comparing assignments on abelian subgroups.
	Furthermore, we prove a sharp threshold for the required degree of the added polynomial phase operators: our construction uses polynomials up to degree $p - 1$, and we show that the group formed by polynomials up to degree $p - 2$ is noncontextual.
\end{abstract}

\section{Introduction}

The Mermin--Peres magic square \cite{Peres,Mermin} is one of the central examples in the foundations of quantum theory.
It gives rise to one of the simplest known proofs of Bell's theorem \cite{Mermin,Aravind-Bell} and the Kochen--Specker theorem (see e.g.~\cite{Budroni-Cabello-Guhne-Kleinmann-Larsson}), and it is one of the simplest examples of pseudotelepathy \cite{Brassard-Broadbent-Tapp}.

Nowadays, the Mermin--Peres square is commonly studied in the language of nonlocal games.
It forms the basis for a class of nonlocal games called \emph{linear constraint system \textup(LCS\textup) games}, which have been studied extensively \cite{Arkhipov,Cleve-Mittal,Cleve-Liu-Slofstra,Kim-Paulsen-Schafhauser,Slofstra-Tsirelson,Slofstra-closed,Coladangelo-Stark,Goldberg,PRSS,Slofstra-Zhang}.
Notably, LCS games were used by Slofstra to prove deep results in quantum information theory, including a negative answer to the strong Tsirelson problem \cite{Slofstra-Tsirelson} and the first proof that the set of quantum correlations is not closed \cite{Slofstra-closed}.

LCS games over $\Z/2\Z$ (the field of two elements) are quite well understood, and most known results are proved in this setting.
By contrast, LCS games over $\Z/d\Z$ for $d$ odd remain poorly understood, since it appears to be much harder to obtain pseudotelepathy here.
For a long time, the following problem has been open:
\begin{question}
    \label{q:main}
    Does there exist an odd integer $d$ and a linear system $Ax = b$ over $\Z/d\Z$ that does not have a classical solution but does have a finite-dimensional operator solution?
\end{question}

For the precise definition of operator solutions, see \cref{sec:gp-sols}.
In essence, \cref{q:main} asks for analogues of the Mermin--Peres magic square over $\Z/d\Z$ with $d$ odd.
A linear system that does not admit a classical solution but does admit a finite-dimensional operator solution is said to admit \emph{pseudotelepathy} \cite{Brassard-Broadbent-Tapp}.
In the world of quantum constraint satisfaction problems, this is sometimes called a \emph{separation of the first kind} \cite{Bulatov-Zivny}.

Several authors tried to construct examples answering \cref{q:main} \cite{Coladangelo-Stark,Qassim-Wallman,Frembs-Okay-Chung,Chung-Okay-Sikora,Bulatov-Zivny,Frembs-Okay-Chung-2}.
On multiple occasions, such examples were announced on arXiv but later retracted due to a mistake.
Some negative results are known as well: a linear system modulo odd $d$ with solutions using generalized Pauli matrices also has a classical solution \cite{Qassim-Wallman}, and other families of finite $p$-groups have been ruled out as well \cite{Frembs-Okay-Chung,Chung-Okay-Sikora}.
Moreover, the special class of graph incidence systems shows a clear separation between $d = 2$ and $d > 2$: for $d = 2$, the graph LCS admits pseudotelepathy if and only if the graph is non-planar, whereas for $d > 2$ a necessary (but not sufficient) condition is that the graph is not projective planar \cite{Slofstra-Zhang,Dobben-Roberson-Negami}.
Since the Mermin--Peres magic square and Mermin's magic pentagram \cite{Mermin-pentagram} can be seen as graph LCS corresponding to the graphs $K_{3,3}$ and $K_5$ (see \cite{Arkhipov}), which are non-planar but still projective planar, it immediately follows that these two graph incidence systems do not admit operator solutions over $\Z/d\Z$ with $d > 2$.

Recently, Chung, Okay and Sikora conjectured a negative answer to \cref{q:main}.

\begin{conjecture}[{\cite[Conjecture~3.16]{Chung-Okay-Sikora}}]
	\label{conj:Chung-Okay-Sikora}
	If $d$ is odd, then any linear system $Ax = b$ over $\Z/d\Z$ which has a finite-dimensional operator solution also has a classical solution.
\end{conjecture}

The first positive results towards \cref{q:main} were obtained by Slofstra and Zhang \cite{Slofstra-Zhang}, who provided a wealth of examples where there is no classical solution but there is a possibly infinite-dimensional operator solution (for arbitrary $d$).
The first positive answer to \cref{q:main}, and thus the first counterexample to \cref{conj:Chung-Okay-Sikora}, was found by the first author and Roberson \cite{Dobben-Roberson-Negami}, who found the first examples for small odd primes.

In this paper, we provide a positive answer to \cref{q:main} for all $d$.

\begin{theorem}[Main result]
	\label{thm:main}
	For every integer $d \geq 2$, there exists a linear system $Ax = b$ over $\Z/d\Z$ that does not have a classical solution but does have a finite-dimensional operator solution.
\end{theorem}

Independent proofs of \cref{thm:main} were announced by Ciardo \cite{Ciardo} and Frembs \cite{Frembs}, all within the same week and each using different techniques.

Our proof is constructive.
For each prime $p$ and integer $d \geq 1$, we construct a finite $p$-group $G_{p,d}$.
We then show that there is a linear system $Ax = b$ over $\Z/p\Z$ that does not have a classical solution but does have a group-valued solution taking values in $G_{p,p-1}$.
From this, \Cref{thm:main} for all $d$ follows as a simple corollary.

The groups $G_{p,d}$ consist of generalized Pauli matrices acting on two $p$-dimensional qudits, but with more diagonal matrices added.
The extra diagonal matrices in $G_{p,d}$ have the form $D_f |z\rangle = \omega^{f(z)} |z\rangle$ acting on the Hilbert space $\mathcal H = \spn\{ |z\rangle : z \in \F_p^2\}$, where $f \in \F_p[x,y]$ ranges over all bivariate polynomials of total degree $\leq d$ (see \cref{sec:construction} for further details).
Setting $d = 1$ gives the generalized Pauli group on two $p$-dimensional qudits, so for $d > 1$ we get a strictly larger group that still bears a strong resemblance to the generalized Pauli group.
To get pseudotelepathy, for $p = 2$ the group $G_{p,p-1} = G_{p,1}$ that we use is exactly the two-qubit Pauli group, but for $p \geq 3$ the group $G_{p,p-1}$ that we use is strictly larger than the generalized Pauli group $G_{p,1}$.

For small $p$, we are able to write down concrete realizations of our result, leading to odd prime cousins of the Mermin--Peres magic square.
The operators in our construction have dimension $p^2$.
This improves over the smallest examples found in \cite{Dobben-Roberson-Negami}, which are $81$-dimensional (for $p = 3$), but Ciardo \cite{Ciardo} and Frembs \cite{Frembs} manage even smaller dimension ($2p$ and $p$, respectively).

\bigskip
As mentioned above, the generalized Pauli group $G_{p,1}$ does not admit pseudotelepathy \cite{Qassim-Wallman}.
This raises the question where the threshold lies that makes pseudotelepathy possible.
Our second main result answers this question exactly:

\begin{theorem}
	\label{thm:threshold}
	For every prime $p$, every linear system $Ax = b$ over $\Z/(p)$ that has a solution in $G_{p,p-2}$ also has a classical solution.
\end{theorem}

In \cref{sec:threshold}, we prove \cref{thm:threshold} more generally for any number of qudits.
This result shows that polynomials of degree $p - 1$ are necessary and sufficient to get a positive answer to \cref{q:main}.
This explains why Pauli matrices are sufficient for $p = 2$ but not for $p > 2$.

\paragraph{Acknowledgements}

We are grateful to William Slofstra for many helpful discussions on LCS games and solution groups.
JvDdB was supported by the Carlsberg Foundation Young Researcher Fellowship CF21-0682 --- ``Quantum Graph Theory'' and by GAČR grant 25-17377S.
RvDdB was funded by NWO grant VI.Veni.212.204 and ERC Horizon grant no.~101042990.
PZ was funded by the European Union (ERC, POCOCOP, 101071674). Views and opinions expressed are however those of the author(s) only and do not necessarily reflect those of the European Union or the European Research Council Executive Agency. Neither the European Union nor the granting authority can be held responsible for them.
Part of the research was carried out while JvDdB and PZ were in residence at the Isaac Newton Institute for Mathematical Sciences in Cambridge, United Kingdom during the research programme ``Quantum Information, Quantum Groups and Operator Algebras'' in November--December 2024, supported by EPSRC grant EP/Z000580/1.
Part of the research was carried out while JvDdB and PZ were in residence at Institut Mittag-Leffler in Djursholm, Sweden during the research semester ``Operator Algebras and Quantum Information'' in the spring of 2026, supported by the Swedish Research Council under grant no.~2021-06594.

\paragraph{AI disclosure}

No AI was used in the research or writing of this paper.

\section{A magic rectangle modulo 3}

Before we dive into our construction and proof, we present a small Mermin--Peres magic rectangle over $\Z/3\Z$ that was constructed from the group $G_{3,2}$.
Here, we have a $5 \times 6$ rectangle filled with $9$-dimensional operator entries, with empty cells containing the identity operator.
The operators within each row or column commute.
The product along each row or column is the identity, except the last row, where the product is $\omega I$ (where $\omega = e^{2\pi i / 3}$).
This represents an operator solution to a linear system with $18$ variables (the non-empty cells) and $11$ equations (the rows and columns), with the condition that the sum of the variables along each row or column is $0 \pmod 3$, except the last row, where the sum if $1 \pmod 3$.
Evidently this linear system has no classical solution, since the sum of all entries of the magic rectangle cannot simultaneously be $0 \pmod 3$ and $1 \pmod 3$.

\begin{center}
	\begin{tikzpicture}[scale=2.4]
		\setlength\arraycolsep{3pt}
		\newcommand{\nil}{\textcolor{gray!70}{0}}
		\def\myscale{.55}
		\def\myomega{\mathbf{\omega}}
		\def\myomegasq{\overline{\mathbf{\omega}}}
		\def\een{\bm{1}}
		\draw (0,0) grid (6,5);
		\begin{scope}[xshift=-5mm,yshift=-5mm]
			\node[scale=\myscale] at (1, 5) {$\begin{pmatrix}
					 \nil & \nil & \nil & 1 & \nil & \nil & \nil & \nil & \nil \\
					 \nil & \nil & \nil & \nil & 1 & \nil & \nil & \nil & \nil \\
					 \nil & \nil & \nil & \nil & \nil & 1 & \nil & \nil & \nil \\
					 \nil & \nil & \nil & \nil & \nil & \nil & 1 & \nil & \nil \\
					 \nil & \nil & \nil & \nil & \nil & \nil & \nil & 1 & \nil \\
					 \nil & \nil & \nil & \nil & \nil & \nil & \nil & \nil & 1 \\
					 1 & \nil & \nil & \nil & \nil & \nil & \nil & \nil & \nil \\
					 \nil & 1 & \nil & \nil & \nil & \nil & \nil & \nil & \nil \\
					 \nil & \nil & 1 & \nil & \nil & \nil & \nil & \nil & \nil
				\end{pmatrix}$};
			\node[scale=\myscale] at (2, 5) {$\begin{pmatrix}
					 \nil & 1 & \nil & \nil & \nil & \nil & \nil & \nil & \nil \\
					 \nil & \nil & 1 & \nil & \nil & \nil & \nil & \nil & \nil \\
					 1 & \nil & \nil & \nil & \nil & \nil & \nil & \nil & \nil \\
					 \nil & \nil & \nil & \nil & 1 & \nil & \nil & \nil & \nil \\
					 \nil & \nil & \nil & \nil & \nil & 1 & \nil & \nil & \nil \\
					 \nil & \nil & \nil & 1 & \nil & \nil & \nil & \nil & \nil \\
					 \nil & \nil & \nil & \nil & \nil & \nil & \nil & 1 & \nil \\
					 \nil & \nil & \nil & \nil & \nil & \nil & \nil & \nil & 1 \\
					 \nil & \nil & \nil & \nil & \nil & \nil & 1 & \nil & \nil
				\end{pmatrix}$};
			\node[scale=\myscale] at (3, 5) {$\begin{pmatrix}
					 \nil & \nil & \nil & \nil & \nil & \nil & \nil & \nil & 1 \\
					 \nil & \nil & \nil & \nil & \nil & \nil & 1 & \nil & \nil \\
					 \nil & \nil & \nil & \nil & \nil & \nil & \nil & 1 & \nil \\
					 \nil & \nil & 1 & \nil & \nil & \nil & \nil & \nil & \nil \\
					 1 & \nil & \nil & \nil & \nil & \nil & \nil & \nil & \nil \\
					 \nil & 1 & \nil & \nil & \nil & \nil & \nil & \nil & \nil \\
					 \nil & \nil & \nil & \nil & \nil & 1 & \nil & \nil & \nil \\
					 \nil & \nil & \nil & 1 & \nil & \nil & \nil & \nil & \nil \\
					 \nil & \nil & \nil & \nil & 1 & \nil & \nil & \nil & \nil
				\end{pmatrix}$};
			\node[scale=\myscale] at (1, 4) {$\begin{pmatrix}
					 \nil & \nil & \nil & \nil & \nil & \nil & 1 & \nil & \nil \\
					 \nil & \nil & \nil & \nil & \nil & \nil & \nil & \myomega & \nil \\
					 \nil & \nil & \nil & \nil & \nil & \nil & \nil & \nil & \myomega \\
					 1 & \nil & \nil & \nil & \nil & \nil & \nil & \nil & \nil \\
					 \nil & \myomega & \nil & \nil & \nil & \nil & \nil & \nil & \nil \\
					 \nil & \nil & \myomega & \nil & \nil & \nil & \nil & \nil & \nil \\
					 \nil & \nil & \nil & 1 & \nil & \nil & \nil & \nil & \nil \\
					 \nil & \nil & \nil & \nil & \myomega & \nil & \nil & \nil & \nil \\
					 \nil & \nil & \nil & \nil & \nil & \myomega & \nil & \nil & \nil
				\end{pmatrix}$};
			\node[scale=\myscale] at (4, 4) {$\begin{pmatrix}
					 \nil & 1 & \nil & \nil & \nil & \nil & \nil & \nil & \nil \\
					 \nil & \nil & 1 & \nil & \nil & \nil & \nil & \nil & \nil \\
					 1 & \nil & \nil & \nil & \nil & \nil & \nil & \nil & \nil \\
					 \nil & \nil & \nil & \nil & \myomegasq & \nil & \nil & \nil & \nil \\
					 \nil & \nil & \nil & \nil & \nil & 1 & \nil & \nil & \nil \\
					 \nil & \nil & \nil & \myomega & \nil & \nil & \nil & \nil & \nil \\
					 \nil & \nil & \nil & \nil & \nil & \nil & \nil & \myomega & \nil \\
					 \nil & \nil & \nil & \nil & \nil & \nil & \nil & \nil & 1 \\
					 \nil & \nil & \nil & \nil & \nil & \nil & \myomegasq & \nil & \nil
				\end{pmatrix}$};
			\node[scale=\myscale] at (5, 4) {$\begin{pmatrix}
					 \nil & \nil & \nil & \nil & \nil & \myomegasq & \nil & \nil & \nil \\
					 \nil & \nil & \nil & 1 & \nil & \nil & \nil & \nil & \nil \\
					 \nil & \nil & \nil & \nil & \myomegasq & \nil & \nil & \nil & \nil \\
					 \nil & \nil & \nil & \nil & \nil & \nil & \nil & \nil & \myomega \\
					 \nil & \nil & \nil & \nil & \nil & \nil & \myomega & \nil & \nil \\
					 \nil & \nil & \nil & \nil & \nil & \nil & \nil & \myomegasq & \nil \\
					 \nil & \nil & 1 & \nil & \nil & \nil & \nil & \nil & \nil \\
					 \myomegasq & \nil & \nil & \nil & \nil & \nil & \nil & \nil & \nil \\
					 \nil & \myomegasq & \nil & \nil & \nil & \nil & \nil & \nil & \nil
				\end{pmatrix}$};
			\node[scale=\myscale] at (2, 3) {$\begin{pmatrix}
					 \nil & \nil & 1 & \nil & \nil & \nil & \nil & \nil & \nil \\
					 1 & \nil & \nil & \nil & \nil & \nil & \nil & \nil & \nil \\
					 \nil & 1 & \nil & \nil & \nil & \nil & \nil & \nil & \nil \\
					 \nil & \nil & \nil & \nil & \nil & \myomega & \nil & \nil & \nil \\
					 \nil & \nil & \nil & \myomega & \nil & \nil & \nil & \nil & \nil \\
					 \nil & \nil & \nil & \nil & \myomega & \nil & \nil & \nil & \nil \\
					 \nil & \nil & \nil & \nil & \nil & \nil & \nil & \nil & \myomega \\
					 \nil & \nil & \nil & \nil & \nil & \nil & \myomega & \nil & \nil \\
					 \nil & \nil & \nil & \nil & \nil & \nil & \nil & \myomega & \nil
				\end{pmatrix}$};
			\node[scale=\myscale] at (5, 3) {$\begin{pmatrix}
					 \nil & \nil & \nil & \nil & \nil & \nil & \nil & \myomegasq & \nil \\
					 \nil & \nil & \nil & \nil & \nil & \nil & \nil & \nil & \myomega \\
					 \nil & \nil & \nil & \nil & \nil & \nil & 1 & \nil & \nil \\
					 \nil & 1 & \nil & \nil & \nil & \nil & \nil & \nil & \nil \\
					 \nil & \nil & \myomega & \nil & \nil & \nil & \nil & \nil & \nil \\
					 \myomegasq & \nil & \nil & \nil & \nil & \nil & \nil & \nil & \nil \\
					 \nil & \nil & \nil & \nil & \myomegasq & \nil & \nil & \nil & \nil \\
					 \nil & \nil & \nil & \nil & \nil & \myomegasq & \nil & \nil & \nil \\
					 \nil & \nil & \nil & \myomegasq & \nil & \nil & \nil & \nil & \nil
				\end{pmatrix}$};
			\node[scale=\myscale] at (6, 3) {$\begin{pmatrix}
					 \nil & \nil & \nil & 1 & \nil & \nil & \nil & \nil & \nil \\
					 \nil & \nil & \nil & \nil & \myomegasq & \nil & \nil & \nil & \nil \\
					 \nil & \nil & \nil & \nil & \nil & \myomega & \nil & \nil & \nil \\
					 \nil & \nil & \nil & \nil & \nil & \nil & 1 & \nil & \nil \\
					 \nil & \nil & \nil & \nil & \nil & \nil & \nil & 1 & \nil \\
					 \nil & \nil & \nil & \nil & \nil & \nil & \nil & \nil & 1 \\
					 1 & \nil & \nil & \nil & \nil & \nil & \nil & \nil & \nil \\
					 \nil & \myomega & \nil & \nil & \nil & \nil & \nil & \nil & \nil \\
					 \nil & \nil & \myomegasq & \nil & \nil & \nil & \nil & \nil & \nil
				\end{pmatrix}$};
			\node[scale=\myscale] at (3, 2) {$\begin{pmatrix}
					 \nil & \nil & \nil & \nil & \myomega & \nil & \nil & \nil & \nil \\
					 \nil & \nil & \nil & \nil & \nil & 1 & \nil & \nil & \nil \\
					 \nil & \nil & \nil & 1 & \nil & \nil & \nil & \nil & \nil \\
					 \nil & \nil & \nil & \nil & \nil & \nil & \nil & 1 & \nil \\
					 \nil & \nil & \nil & \nil & \nil & \nil & \nil & \nil & \myomega \\
					 \nil & \nil & \nil & \nil & \nil & \nil & 1 & \nil & \nil \\
					 \nil & 1 & \nil & \nil & \nil & \nil & \nil & \nil & \nil \\
					 \nil & \nil & 1 & \nil & \nil & \nil & \nil & \nil & \nil \\
					 \myomega & \nil & \nil & \nil & \nil & \nil & \nil & \nil & \nil
				\end{pmatrix}$};
			\node[scale=\myscale] at (4, 2) {$\begin{pmatrix}
					 \nil & \nil & 1 & \nil & \nil & \nil & \nil & \nil & \nil \\
					 1 & \nil & \nil & \nil & \nil & \nil & \nil & \nil & \nil \\
					 \nil & 1 & \nil & \nil & \nil & \nil & \nil & \nil & \nil \\
					 \nil & \nil & \nil & \nil & \nil & 1 & \nil & \nil & \nil \\
					 \nil & \nil & \nil & \myomegasq & \nil & \nil & \nil & \nil & \nil \\
					 \nil & \nil & \nil & \nil & \myomega & \nil & \nil & \nil & \nil \\
					 \nil & \nil & \nil & \nil & \nil & \nil & \nil & \nil & \myomegasq \\
					 \nil & \nil & \nil & \nil & \nil & \nil & 1 & \nil & \nil \\
					 \nil & \nil & \nil & \nil & \nil & \nil & \nil & \myomega & \nil
				\end{pmatrix}$};
			\node[scale=\myscale] at (6, 2) {$\begin{pmatrix}
					 \nil & \nil & \nil & \nil & \nil & \nil & 1 & \nil & \nil \\
					 \nil & \nil & \nil & \nil & \nil & \nil & \nil & 1 & \nil \\
					 \nil & \nil & \nil & \nil & \nil & \nil & \nil & \nil & \myomegasq \\
					 1 & \nil & \nil & \nil & \nil & \nil & \nil & \nil & \nil \\
					 \nil & \myomegasq & \nil & \nil & \nil & \nil & \nil & \nil & \nil \\
					 \nil & \nil & 1 & \nil & \nil & \nil & \nil & \nil & \nil \\
					 \nil & \nil & \nil & 1 & \nil & \nil & \nil & \nil & \nil \\
					 \nil & \nil & \nil & \nil & \myomega & \nil & \nil & \nil & \nil \\
					 \nil & \nil & \nil & \nil & \nil & \myomega & \nil & \nil & \nil
				\end{pmatrix}$};
			\node[scale=\myscale] at (1, 1) {$\begin{pmatrix}
					 1 & \nil & \nil & \nil & \nil & \nil & \nil & \nil & \nil \\
					 \nil & \myomegasq & \nil & \nil & \nil & \nil & \nil & \nil & \nil \\
					 \nil & \nil & \myomegasq & \nil & \nil & \nil & \nil & \nil & \nil \\
					 \nil & \nil & \nil & 1 & \nil & \nil & \nil & \nil & \nil \\
					 \nil & \nil & \nil & \nil & \myomegasq & \nil & \nil & \nil & \nil \\
					 \nil & \nil & \nil & \nil & \nil & \myomegasq & \nil & \nil & \nil \\
					 \nil & \nil & \nil & \nil & \nil & \nil & 1 & \nil & \nil \\
					 \nil & \nil & \nil & \nil & \nil & \nil & \nil & \myomegasq & \nil \\
					 \nil & \nil & \nil & \nil & \nil & \nil & \nil & \nil & \myomegasq
				\end{pmatrix}$};
			\node[scale=\myscale] at (2, 1) {$\begin{pmatrix}
					 1 & \nil & \nil & \nil & \nil & \nil & \nil & \nil & \nil \\
					 \nil & 1 & \nil & \nil & \nil & \nil & \nil & \nil & \nil \\
					 \nil & \nil & 1 & \nil & \nil & \nil & \nil & \nil & \nil \\
					 \nil & \nil & \nil & \myomegasq & \nil & \nil & \nil & \nil & \nil \\
					 \nil & \nil & \nil & \nil & \myomegasq & \nil & \nil & \nil & \nil \\
					 \nil & \nil & \nil & \nil & \nil & \myomegasq & \nil & \nil & \nil \\
					 \nil & \nil & \nil & \nil & \nil & \nil & \myomegasq & \nil & \nil \\
					 \nil & \nil & \nil & \nil & \nil & \nil & \nil & \myomegasq & \nil \\
					 \nil & \nil & \nil & \nil & \nil & \nil & \nil & \nil & \myomegasq
				\end{pmatrix}$};
			\node[scale=\myscale] at (3, 1) {$\begin{pmatrix}
					 \myomegasq & \nil & \nil & \nil & \nil & \nil & \nil & \nil & \nil \\
					 \nil & 1 & \nil & \nil & \nil & \nil & \nil & \nil & \nil \\
					 \nil & \nil & 1 & \nil & \nil & \nil & \nil & \nil & \nil \\
					 \nil & \nil & \nil & 1 & \nil & \nil & \nil & \nil & \nil \\
					 \nil & \nil & \nil & \nil & \myomegasq & \nil & \nil & \nil & \nil \\
					 \nil & \nil & \nil & \nil & \nil & 1 & \nil & \nil & \nil \\
					 \nil & \nil & \nil & \nil & \nil & \nil & 1 & \nil & \nil \\
					 \nil & \nil & \nil & \nil & \nil & \nil & \nil & 1 & \nil \\
					 \nil & \nil & \nil & \nil & \nil & \nil & \nil & \nil & \myomegasq
				\end{pmatrix}$};
			\node[scale=\myscale] at (4, 1) {$\begin{pmatrix}
					 1 & \nil & \nil & \nil & \nil & \nil & \nil & \nil & \nil \\
					 \nil & 1 & \nil & \nil & \nil & \nil & \nil & \nil & \nil \\
					 \nil & \nil & 1 & \nil & \nil & \nil & \nil & \nil & \nil \\
					 \nil & \nil & \nil & \myomegasq & \nil & \nil & \nil & \nil & \nil \\
					 \nil & \nil & \nil & \nil & \myomegasq & \nil & \nil & \nil & \nil \\
					 \nil & \nil & \nil & \nil & \nil & \myomegasq & \nil & \nil & \nil \\
					 \nil & \nil & \nil & \nil & \nil & \nil & \myomegasq & \nil & \nil \\
					 \nil & \nil & \nil & \nil & \nil & \nil & \nil & \myomegasq & \nil \\
					 \nil & \nil & \nil & \nil & \nil & \nil & \nil & \nil & \myomegasq
				\end{pmatrix}$};
			\node[scale=\myscale] at (5, 1) {$\begin{pmatrix}
					 \myomegasq & \nil & \nil & \nil & \nil & \nil & \nil & \nil & \nil \\
					 \nil & 1 & \nil & \nil & \nil & \nil & \nil & \nil & \nil \\
					 \nil & \nil & 1 & \nil & \nil & \nil & \nil & \nil & \nil \\
					 \nil & \nil & \nil & 1 & \nil & \nil & \nil & \nil & \nil \\
					 \nil & \nil & \nil & \nil & 1 & \nil & \nil & \nil & \nil \\
					 \nil & \nil & \nil & \nil & \nil & \myomegasq & \nil & \nil & \nil \\
					 \nil & \nil & \nil & \nil & \nil & \nil & 1 & \nil & \nil \\
					 \nil & \nil & \nil & \nil & \nil & \nil & \nil & \myomegasq & \nil \\
					 \nil & \nil & \nil & \nil & \nil & \nil & \nil & \nil & 1
				\end{pmatrix}$};
			\node[scale=\myscale] at (6, 1) {$\begin{pmatrix}
					 1 & \nil & \nil & \nil & \nil & \nil & \nil & \nil & \nil \\
					 \nil & \myomegasq & \nil & \nil & \nil & \nil & \nil & \nil & \nil \\
					 \nil & \nil & \myomegasq & \nil & \nil & \nil & \nil & \nil & \nil \\
					 \nil & \nil & \nil & 1 & \nil & \nil & \nil & \nil & \nil \\
					 \nil & \nil & \nil & \nil & \myomegasq & \nil & \nil & \nil & \nil \\
					 \nil & \nil & \nil & \nil & \nil & \myomegasq & \nil & \nil & \nil \\
					 \nil & \nil & \nil & \nil & \nil & \nil & 1 & \nil & \nil \\
					 \nil & \nil & \nil & \nil & \nil & \nil & \nil & \myomegasq & \nil \\
					 \nil & \nil & \nil & \nil & \nil & \nil & \nil & \nil & \myomegasq
				\end{pmatrix}$};
		\end{scope}
	\end{tikzpicture}
\end{center}

\section{Group-valued solutions to linear systems}
\label{sec:gp-sols}

Let $d\ge2$, let $A=(A_{ij})\in(\Z/d\Z)^{m\times n}$, and let
$b\in(\Z/d\Z)^m$. Write $\omega_d=e^{2\pi i/d}$.

\begin{definition}\label{def:classical}
A \emph{classical solution} of $Ax=b$ is a vector
$x\in(\Z/d\Z)^n$ satisfying
\[
 \sum_{j=1}^n A_{ij}x_j=b_i,\qquad 1\le i\le m.
\]
\end{definition}

Encoding a scalar $a\in\Z/d\Z$ by $\omega_d^a$ turns addition into
multiplication. This motivates the following definition.

\begin{definition}\label{def:operator-solution}
An \emph{operator solution} of $Ax=b$ consists of unitary operators
$U_1,\ldots,U_n$ on a common nonzero complex Hilbert space
$\mathcal H$, satisfying:
\begin{enumerate}[label=\textup{(\roman*)}]
\item $U_j^d=I$ for every $j$;
\item $U_jU_k=U_kU_j$ whenever $A_{ij}\ne0$ and $A_{ik}\ne0$
for some row $i$;
\item for every row $i$,
\begin{equation}\label{eq:operator-solution-row}
 \prod_{j=1}^n U_j^{A_{ij}}=\omega_d^{b_i}I.
\end{equation}
\end{enumerate}
The solution is \emph{finite-dimensional} if $\dim\mathcal H<\infty$.
\end{definition}

The condition $U_j^d=I$ requires the order of $U_j$ to divide $d$;
it need not equal $d$. Exponents are interpreted modulo $d$, and
the commutation condition makes each row product independent of its
order. Classical solutions correspond exactly to one-dimensional
operator solutions, via $U_j=\omega_d^{x_j}$.

\begin{definition}\label{def:group-solution}
A \emph{group-valued solution} of $Ax=b$ consists of a group $G$,
a distinguished central element $J\in G$ of order $d$, and elements
$h_1,\ldots,h_n\in G$ such that
\[
 h_j^d=e_G,\quad 1\le j\le n,\qquad
 \prod_{j=1}^n h_j^{A_{ij}}=J^{b_i},\quad 1\le i\le m,
\]
with $h_jh_k=h_kh_j$ whenever $A_{ij}\ne0$ and $A_{ik}\ne0$
for some row $i$. Here $e_G$ denotes the identity of $G$.
For fixed $G$ and $J$, we also call this a \emph{$G$-valued solution}.
\end{definition}

The element $J$ represents the scalar $1 \in \Z/d\Z$. Since it has order exactly
$d$, classical solutions correspond to the $G$-valued solutions with
$h_j\in\langle J\rangle$, via $h_j=J^{x_j}$.
If $\rho\colon G\to U(\mathcal H)$ is a unitary representation on a
nonzero Hilbert space with $\rho(J)=\omega_d I$, then
$U_j=\rho(h_j)$ is an operator solution:
the representation preserves powers, commutation, and row products.

We now construct a linear system from any finite group $G$ with a
distinguished central element $J$ of order $d$. Put
\[
 \OO_d(G)=\{u\in G:u^d=e_G\}.
\]
If $u,v\in\OO_d(G)$ commute, then
$(uv)^d=u^dv^d=e_G$, so $uv\in\OO_d(G)$.

\begin{construction}\label{cons:finite-system}
Introduce a variable $x_u$ for every $u\in\OO_d(G)$ and consider the linear system consisting of the following equations over $\Z/d\Z$:
\begin{equation}\label{eq:system}
 x_u+x_v-x_{uv}=0,
 \quad u,v\in\OO_d(G),\ uv=vu,\qquad x_J=1.
\end{equation}
When variable indices coincide, their coefficients are added modulo $d$.
\end{construction}

\begin{proposition}\label{prop:system-reduction}
The finite system in \cref{cons:finite-system} has a $G$-valued solution.
A classical solution is exactly a function
$\lambda\colon\OO_d(G)\to\Z/d\Z$ satisfying
\begin{equation}\label{eq:additivity}
 \lambda(uv)=\lambda(u)+\lambda(v),\quad u,v\in\OO_d(G),\ uv=vu,
 \qquad \lambda(J)=1.
\end{equation}
\end{proposition}

\begin{proof}
Assign $u$ to $x_u$. For each commuting pair $u,v$, the elements
$u,v,uv$ commute pairwise and satisfy $uv(uv)^{-1}=e_G$.
The final equation becomes $J=J$, so this is a $G$-valued solution.
For a classical assignment, setting $\lambda(u)=x_u$ identifies
\eqref{eq:system} with \eqref{eq:additivity}.
\end{proof}

For prime moduli $d=p$, we identify $\Z/p\Z$ with $\F_p$.
Every map $\lambda\colon\OO_p(G)\to\F_p$ additive on commuting pairs satisfies
\[
 \lambda(e_G)=0,\qquad \lambda(u^r)=r\lambda(u),
 \quad u\in\OO_p(G),\ r\in\F_p,
\]
since the powers of $u$ commute.
To rule out classical solutions, it will suffice to show that every
function $\lambda\colon\OO_p(G)\to\F_p$ additive on commuting pairs
has $\lambda(J)=0$.

\section{The group and its representation as Pauli-like matrices}
\label{sec:construction}

Throughout the remainder of this paper, let $p \geq 2$ be a fixed prime number, let $\omega = \omega_p = e^{2\pi i/p}$, let $\F_p$ be the finite field of order $p$, and let $E$ be the additive group of the $2$-dimensional vector space $\F_p^2$ over $\F_p$.
Write $e_1=(1,0)$ and $e_2=(0,1)$ for its standard basis, and $x,y$ for the coordinate functions.

\begin{definition}
	Write $\mathcal A := \F_p[x,y]_{\leq (p-1,p-1)} = \spn\{ x^i y^j \mid i,j \in \{0,1,\ldots,p-1\} \}$ for the $\F_p$-vector space of bivariate polynomials of bidegree at most $(p-1,p-1)$.
	The \emph{evaluation map} $\ev\colon \F_p[x,y] \to \F_p^{E} = \{\text{functions $f\colon E \to \F_p$}\}$ sends a polynomial $Q \in \F_p[x,y]$ to the function $\ev(Q)\colon \F_p^2 \to \F_p$, $(r,s) \mapsto Q(r,s)$ it defines.
	Its kernel is the ideal $(x^p - x , y^p - y)$, and it restricts to an isomorphism $\mathcal A \stackrel{\sim}{\to} \F_p^E$.
	
	For $d \geq 0$, write $\mathcal P_d := \spn\{ x^i y^j \mid i + j \leq d \} \subseteq \F_p[x,y]$ for the space of bivariate polynomials of total degree at most $d$, and write $\mathcal A_d := \mathcal P_d \cap \mathcal A = \spn\{ x^i y^j \mid i,j \in \{0,1,\ldots,p-1\}, i + j \leq d\}$.
\end{definition}

Via $\ev$, we identify $\mathcal A$ with the space of functions on $E$.
Under this identification, $\mathcal A_d$ consists exactly of the functions
with a polynomial representative of total degree at most $d$: reducing
modulo $x^p-x$ and $y^p-y$ does not increase total degree.
In particular, identities in $\mathcal A$ are identities of functions,
so $x^p=x$ and $y^p=y$.

\begin{definition}
	\label{def:translation}
	For $a \in E$, define the \emph{translation operator} by $(T_af)(z) = f(z + a)$, and the \emph{finite difference in direction $a$} by $(\Delta_af)(z) = f(z + a) - f(z)$.
\end{definition}
	
The translation operator satisfies
\[ T_a T_b = T_{a+b} , \qquad T_0 = \id , \qquad T_a^{-1} = T_{-a}. \]
Moreover, since
\[ (x + r)^k (y + s)^\ell = \Bigg( \sum_{i=0}^k \binom{k}{i} r^{k-i} x^i \Bigg) \Bigg( \sum_{j=0}^\ell \binom{\ell}{j} s^{\ell-j} y^j \Bigg) = \sum_{i=0}^k \sum_{j=0}^\ell \binom{k}{i} \binom{\ell}{j} r^{k-i} s^{\ell-j} x^i y^j  , \]
we see that $T_a$ preserves the $x$-degree, $y$-degree, and total degree of each polynomial.
The highest-degree terms cancel in $T_af-f$, so
$\Delta_a(\mathcal A_{r+1})\subseteq\mathcal A_r$ for every $r\ge0$.

\begin{definition}
	For $0 \leq d \leq p - 1$, define the group $G_{p,d}$ as the semidirect product $G_{p,d} = \mathcal A_d \rtimes E$, where $E$ acts on $\mathcal A_d$ by translation.
\end{definition}

In other words, elements of $G_{p,d}$ are tuples $(f,a) \in \mathcal A_d \times E$, with multiplication given by
\[ (f,a)(h,b) = (f + T_ah , a + b). \]
This is a finite group with identity $(0,0)$. The subgroup
$\mathcal A_d\times\{0\}$ is additive, and $J=(1,0)$ is central of order $p$.
The multiplication rule also shows that $(h,0)$ commutes with $(f,a)$
if and only if $\Delta_a h=0$.

For $g\in\mathcal A_{d+1}$, define
\begin{equation}\label{eq:tau}
 \tau_g\colon E\to G_{p,d},\qquad \tau_g(a)=(\Delta_a g,a).
\end{equation}
This is well-defined because $\Delta_a g\in\mathcal A_d$. Moreover,
\[
 \Delta_a g+T_a\Delta_b g=\Delta_{a+b}g,
 \qquad \tau_g(a)\tau_g(b)=\tau_g(a+b).
\]
Thus $\tau_g$ is a homomorphism, injective because its second coordinate
is $a$. Its image is an abelian subgroup isomorphic to $E$, and
$\tau_g(a)^p=\tau_g(pa)=\tau_g(0)=(0,0)$, so
$\tau_g(E)\subseteq\OO_p(G_{p,d})$.

\begin{lemma}\label{lem:subgroup-linearity}
Let $\lambda\colon\OO_p(G_{p,d})\to\F_p$ be additive on commuting pairs.
For every $g\in\mathcal A_{d+1}$, the map
$\lambda\circ\tau_g\colon E\to\F_p$ is $\F_p$-linear. In particular,
\[
 \lambda(\tau_g(re_1+se_2))
 =r\lambda(\tau_g(e_1))+s\lambda(\tau_g(e_2)),
 \qquad r,s\in\F_p.
\]
\end{lemma}

\begin{proof}
For $a,b\in E$, the homomorphism property and commuting additivity give
\[
 \lambda(\tau_g(a+b))
 =\lambda(\tau_g(a)\tau_g(b))
 =\lambda(\tau_g(a))+\lambda(\tau_g(b)).
\]
Repeated addition gives compatibility with scalars in $\F_p$.
\end{proof}

\paragraph{The representation.}
The group $G_{p,d}$ has a faithful representation as Pauli-like matrices.
Let $\mathcal H = \C^E$ have orthonormal basis $\{|z\rangle:z\in E\}$.
For $f \in \F_p^E$ and $a \in E$, define
\[  D_f |z\rangle = \omega^{f(z)} |z\rangle , \qquad S_a |z\rangle = |z-a\rangle , \qquad \rho(f, a)  = D_f S_a. \]

\begin{proposition}
    \label{prop:matrix-realization}
    The map $\rho\colon G_{p,d} \to U(\mathcal H)$ is a faithful unitary representation
    with $\rho(J)=\omega I$. Moreover, for $g\in\mathcal A_{d+1}$ and $a\in E$,
    \begin{equation}
        \rho(\tau_g(a)) = D_g^\dagger S_a D_g. \label{eq:represented-conjugation}
    \end{equation}
\end{proposition}
\begin{proof}
    $D_f$ is diagonal and unitary and $S_a$ is a permutation matrix, so $\rho(f,a)$ is unitary for all $(f,a) \in G_{p,d}$.
    Note that $D_f^\dagger = D_{-f}$ and $S_a^\dagger = S_{-a}$.
    Applying the definitions to basis vectors, we see that $S_a D_h S_a^\dagger = D_{T_ah}$.
    Also $D_fD_h=D_{f+h}$ and $S_aS_b=S_{a+b}$, so
    \[ D_fS_aD_hS_b=D_{f+T_ah}S_{a+b}. \]
    This proves that $\rho$ is a homomorphism.
    Its action on basis vectors is
    \[
     \rho(f,a)|z\rangle=\omega^{f(z-a)}|z-a\rangle.
    \]
    This determines $a$ and every value of $f$, since the $p$ powers of
    $\omega$ are distinct. Injectivity of evaluation on $\mathcal A_d$
    therefore makes $\rho$ faithful. Also $\rho(J)=D_1=\omega I$.
    Finally,
    \[
     D_g^\dagger S_aD_g=D_{-g}D_{T_ag}S_a
     =D_{\Delta_a g}S_a=\rho(\tau_g(a)). \qedhere
    \]
\end{proof}

\begin{remark}
	We note that $G_{p,1}$ is the generalized Pauli group acting on two $p$-dimensional qudits.
	Indeed, consider the \emph{generalized Pauli matrices} $X$ and $Z$ acting on $\C^p = \spn\{ |z\rangle : z \in \Z/p\Z \}$, given by $X |z\rangle = |z + 1\rangle$ and $Z |z\rangle = \omega^z |z\rangle$.
	The representation $\rho$ from \cref{prop:matrix-realization} gives an isomorphism between $G_{p,1}$ and the two-qudit generalized Pauli group as follows:
	\[ \rho(ax + by + c , (\begin{smallmatrix} r \\ s \end{smallmatrix})) = \omega^c (Z^a \otimes Z^b)(X^{-r} \otimes X^{-s}). \]
	When $\deg(f) > 1$, the operator $\rho(f,0) = D_f$ is a diagonal matrix on $\mathcal H = \C^p \otimes \C^p$ which still only has powers of $\omega$ on the diagonal, but which does not decompose as $\omega^c (Z^a \otimes Z^b)$.
	Adding these extra diagonal matrices turns out to be the crucial missing ingredient needed to get Mermin--Peres magic rectangles modulo odd primes.
	
	Groups of generalized Pauli matrices with extra diagonal matrices added were also considered in \cite{Chung-Okay-Sikora}, but there it was proved that no pseudotelepathy could be derived from those groups.
	Our construction differs from the one in \cite{Chung-Okay-Sikora} by also adding two-qudit diagonal matrices that do not decompose as tensor products of single qudit matrices, whereas in \cite{Chung-Okay-Sikora} only single-qudit diagonal matrices are added.
	In our framework, this would correspond to adding polynomials of the form $f(x) + g(y)$, which still leaves out multivariate polynomials such as $xy$.
\end{remark}

\section{Proof of the main theorem}

In this section, we prove \cref{thm:main}.
We first treat prime moduli.
For $p = 2$, the Mermin--Peres magic square provides an example, so assume henceforth that $p$ is an odd prime.
Put $G=G_{p,p-1}$ and $J=(1,0) \in G$, as before.
We show that every function $\lambda\colon\OO_p(G)\to\F_p$
that is additive on commuting pairs satisfies $\lambda(J)=0$.
By \cref{prop:system-reduction}, this rules out a classical solution
of the system in \cref{cons:finite-system}.
The key is to compare the restrictions of $\lambda$ to the abelian
subgroups $\tau_g(E)$ as the polynomial parameter $g$ varies.

By \cref{lem:subgroup-linearity}, each $\lambda\circ\tau_g$ is
linear in the translation direction. Its change from $g=0$ is therefore
determined by two coordinates, which we encode in a map $\Phi(g)$.
Commutation relations linking different subgroups give scalar conditions
on the increments $\Phi(g+k)-\Phi(g)$. The following lemma turns these
conditions into linearity of $\Phi$ on a suitable parameter space.
We take $W=\spn\{k_t:t\in\F_p\}\le\mathcal A_p$, where $k_t=x(y-tx)^{p-1}$.
The lemma applies to the increment spaces $\spn\{x,k_t\}$ in
directions $(1,t)$ and $\spn\{x,x^{p-1}y\}$ in direction $e_2$:
for every allowed increment $k$ in direction $u$, the function
$\Delta_u k$ is invariant under translation by $u$. This makes
$(\Delta_u k,0)$ commute with $\tau_g(u)$, yielding the required
scalar condition. These spaces also satisfy the spanning condition.
Linearity then lets us sum the comparisons for the $k_t$.
Using the explicit formulas for these polynomial sums, we cancel all
terms except $\lambda(J)$.

\begin{lemma}[A linearity criterion]\label{lem:linearity}
Let $\mathbb K$ be a prime field, let $W$ be a vector space over $\mathbb K$, and let
$\mathcal D\subseteq\mathbb K^2$ have at least two elements.
Assume that every two distinct elements of $\mathcal D$ are linearly independent.
Suppose subspaces $K_u\le W$, indexed by $u\in\mathcal D$, satisfy
\[
 \spn\left(\bigcup_{v\in\mathcal D\setminus\{u\}}K_v\right)=W,
 \qquad u\in\mathcal D.
\]
If $\Phi\colon W\to\mathbb K^2$ satisfies $\Phi(0)=0$ and
$u\cdot(\Phi(g+k)-\Phi(g))$ is independent of $g$ whenever
$k\in K_u$, then $\Phi$ is $\mathbb K$-linear.
\end{lemma}

\begin{proof}
For distinct $u,v\in\mathcal D$, $k\in K_u$, and $h\in K_v$,
subtracting the hypothesis at $g+h$ and $g$ gives
\[
 u\cdot\bigl(\Phi(g+h+k)-\Phi(g+h)-\Phi(g+k)+\Phi(g)\bigr)=0.
\]
Interchanging $(u,k)$ and $(v,h)$ gives the same equality with $v$
in place of $u$. Since $u,v$ are linearly independent,
\[
 \Phi(g+h+k)-\Phi(g+h)=\Phi(g+k)-\Phi(g).
\]
For fixed $k\in K_u$, the difference $\Phi(g+k)-\Phi(g)$ is therefore
unchanged when any element of $K_v$, $v\ne u$, is added to $g$.
Since these are subspaces spanning $W$, any $g\in W$ admits a finite
decomposition
\[
 g=g_1+\cdots+g_m,
 \qquad g_i\in K_{v_i},\quad v_i\ne u.
\]
Applying the preceding equality successively to the shifts
$g_m,g_{m-1},\ldots,g_1$ gives
\[
\begin{aligned}
 \Phi(g+k)-\Phi(g)
 &=\Phi(g_1+\cdots+g_m+k)-\Phi(g_1+\cdots+g_m)\\
 &=\Phi(g_1+\cdots+g_{m-1}+k)-\Phi(g_1+\cdots+g_{m-1})\\
 &\mathrel{\phantom{=}}\vdots\\
 &=\Phi(g_1+k)-\Phi(g_1)\\
 &=\Phi(k)-\Phi(0)\\
 &=\Phi(k).
\end{aligned}
\]
Thus
\[
 \Phi(g+k)=\Phi(g)+\Phi(k),
 \qquad g\in W,\ k\in K_u,\ u\in\mathcal D.
\]
Now let $k\in W$ be arbitrary. Since the subspaces $K_u$ span $W$,
write
\[
 k=k_1+\cdots+k_n,
 \qquad k_i\in K_{u_i},\quad u_i\in\mathcal D.
\]
Applying the identity just proved successively to $k_n,k_{n-1},\ldots,k_1$
gives, for every $g\in W$,
\[
\begin{aligned}
 \Phi(g+k)
 &=\Phi(g+k_1+\cdots+k_n)\\
 &=\Phi(g+k_1+\cdots+k_{n-1})+\Phi(k_n)\\
 &\mathrel{\phantom{=}}\vdots\\
 &=\Phi(g+k_1)+\Phi(k_2)+\cdots+\Phi(k_n)\\
 &=\Phi(g)+\Phi(k_1)+\cdots+\Phi(k_n).
\end{aligned}
\]
In particular, taking $g=0$ gives
\[
 \Phi(k)=\Phi(0)+\Phi(k_1)+\cdots+\Phi(k_n)
        =\Phi(k_1)+\cdots+\Phi(k_n).
\]
Substituting this into the preceding calculation yields
\[
 \Phi(g+k)=\Phi(g)+\Phi(k),
 \qquad g,k\in W.
\]
Thus $\Phi$ is additive. Since $\mathbb K$ is a prime field,
additivity implies $\mathbb K$-linearity.
\end{proof}

\begin{theorem}\label{thm:vanishing}
Let $p$ be an odd prime, $G=G_{p,p-1}$, and $J=(1,0)$.
If $\lambda\colon\OO_p(G)\to\F_p$ satisfies
\[
 \lambda(uv)=\lambda(u)+\lambda(v),
 \qquad u,v\in\OO_p(G),\ uv=vu,
\]
then $\lambda(J)=0$.
\end{theorem}

\begin{proof}
\emph{Step 1: derive the comparison equations.}
Let $W\le\mathcal A_p$ and let $K_u\le W$, indexed by
$u\in\mathcal D\subseteq E$, satisfy $\Delta_u^2k=0$ for $k\in K_u$.
We first derive comparison equations valid for any such choice of $W$,
$\mathcal D$ and $K_u$. In Step~2 we will choose these spaces and directions
explicitly, ensuring that they also satisfy the spanning and independence
conditions of \cref{lem:linearity}.

For $g\in W$, define $\Phi(g)=(F(g),H(g))$ by
\[
 F(g)=\lambda(\tau_g(e_1))-\lambda(\tau_0(e_1)),\qquad
 H(g)=\lambda(\tau_g(e_2))-\lambda(\tau_0(e_2)).
\]
Then $\Phi(0)=0$, and \cref{lem:subgroup-linearity} gives
\begin{equation}\label{eq:coordinates}
 u\cdot\Phi(g)=\lambda(\tau_g(u))-\lambda(\tau_0(u)),
 \qquad u\in E.
\end{equation}
For $k\in K_u$, put $d=\Delta_uk\in\mathcal A_{p-1}$.
The assumption $\Delta_u^2k=0$ gives
\[
 T_ud-d=\Delta_ud=\Delta_u^2k=0,
\]
so $T_ud=d$. By the group law,
\[
 (d,0)\tau_g(u)=(d+\Delta_ug,u)
 =(\Delta_ug+T_ud,u)=\tau_g(u)(d,0).
\]
Thus $(d,0)$ and $\tau_g(u)$ commute. Since
$\tau_{g+k}(u)=(d,0)\tau_g(u)$, commuting additivity of $\lambda$
and \eqref{eq:coordinates} give
\begin{equation}\label{eq:comparison}
\begin{aligned}
 u\cdot\bigl(\Phi(g+k)-\Phi(g)\bigr)
 &=\lambda(\tau_{g+k}(u))-\lambda(\tau_g(u))\\
 &=\lambda((\Delta_u k,0)),\qquad g\in W,\ k\in K_u.
\end{aligned}
\end{equation}
Since $\mathcal A_{p-1}\times\{0\}$ is an additive subgroup, the map
\[
 L(f)=\lambda((f,0)),\qquad f\in\mathcal A_{p-1}
\]
is $\F_p$-linear, and we need to show that $L(1)=0$.

\smallskip
\noindent\emph{Step 2: choose the parameter spaces.}
For $t\in\F_p$, put
\[
 u_t=(1,t),\qquad k_t=x(y-tx)^{p-1},\qquad
 d_t=(y-tx)^{p-1},\qquad h=x^{p-1}y.
\]
The identities $\binom{p-1}{i}=(-1)^i$ in $\F_p$ give
\[
 (y-tx)^{p-1}=\sum_{i=0}^{p-1}t^ix^iy^{p-1-i}.
\]
For $0\le j\le p$, the sum $\sum_{t\in\F_p}t^j$ is $-1$ if
$j=p-1$ and $0$ otherwise, since $p$ is odd.
Consequently, as identities of functions,
\begin{equation}\label{eq:sums}
 \sum_t k_t=-x,\qquad
 \sum_t t k_t=-h,\qquad
 \sum_t d_t=-x^{p-1}.
\end{equation}

Set $W=\spn\{k_t:t\in\F_p\}\le\mathcal A_p$; it contains $x,h$ by
\eqref{eq:sums}. For the directions
$\mathcal D=\{u_t:t\in\F_p\}\cup\{e_2\}$, put
\[
 K_{u_t}=\spn\{x,k_t\},\qquad K_{e_2}=\spn\{x,h\}.
\]
These subspaces span $W$ after any one is omitted.
Indeed, if $K_{u_t}$ is omitted, recover $k_t$ from
$k_t=-x-\sum_{s\ne t}k_s$; if $K_{e_2}$ is omitted, every $k_t$
remains. Also $\Delta_u^2k=0$ for $k\in K_u$, since
\[
 \Delta_{u_t}x=1,\quad \Delta_{u_t}k_t=d_t,\quad
 \Delta_{u_t}d_t=0,\qquad
 \Delta_{e_2}x=0,\quad \Delta_{e_2}h=x^{p-1}.
\]

\smallskip
\noindent\emph{Step 3: conclude that $\lambda(J)=0$.}
The directions in $\mathcal D$ are pairwise linearly independent, so
\cref{lem:linearity}, with $\mathbb K=\F_p$, now shows that $F,H$
are linear.

Apply \eqref{eq:comparison} at $g=0$ to $(u,k)=(u_t,k_t)$:
\[
 F(k_t)+tH(k_t)=L(d_t).
\]
Summing over $t$ using \eqref{eq:sums} and linearity of $F$ and $H$ gives
\[
 F(x)+H(h)=L(x^{p-1}).
\]
Finally, the comparisons for $(e_1,x)$ and $(e_2,h)$ give
$F(x)=L(1)$ and $H(h)=L(x^{p-1})$. Hence $\lambda(J)=L(1)=0$.
\end{proof}

\begin{proof}[Proof of \cref{thm:main}]
First let $d = p$ be prime.
If $p = 2$, the Mermin--Peres magic square suffices.
If $p > 2$, apply \cref{cons:finite-system} with $d=p$, $G=G_{p,p-1}$, and $J=(1,0)$.
By \cref{prop:system-reduction}, a classical solution would define a function
$\lambda\colon\OO_p(G)\to\F_p$ that is additive on commuting
pairs and satisfies $\lambda(J)=1$, contradicting
\cref{thm:vanishing}. Applying the representation $\rho$ from
\cref{prop:matrix-realization} gives an operator solution on
$\mathcal H=\C^E$, since $\rho(J)=\omega I$.

Now let $d\ge2$ be arbitrary, choose a prime divisor $p$ of $d$, and put $m=d/p$.
Let $Ax=b$ be a system over $\F_p$ obtained above, with operator solution
$U_1,\ldots,U_n$.
Choose coefficient representatives in $\{0,\ldots,p-1\}$ and consider
the system over $\Z/d\Z$
\[
 Ay=mb,\qquad py=0,
\]
where the second equation is imposed on each variable.
Any classical solution has the form $y=mx$ for a unique vector
$x$ over $\F_p$, and then $Ay=mb$ implies $Ax=b$, a contradiction.
The same operators solve this new system: $U_j^p=I$ gives both the
added relations and $U_j^d=I$, the chosen coefficient representatives
preserve the support of each row, and $\omega_p=\omega_d^{\,m}$.
This proves the result for every $d\ge2$.
\end{proof}

\section{Polynomials of degree \texorpdfstring{$p-2$}{p-2} are noncontextual}
\label{sec:threshold}

In this section, we prove \cref{thm:threshold}, which shows that the construction fails if we only allow polynomials up to degree $p - 2$.
This proves a sharp threshold: to construct Mermin--Peres magic rectangles, polynomials of degree $p - 1$ are necessary and sufficient.

We prove this result in a more general setting, where we generalize the definitions from \cref{sec:construction} to an arbitrary number of variables, so the group acts on an arbitrary number of qudits.

\begin{definition}
	\label{def:polys-q}
	Let $p$ be a fixed prime, let $q \geq 1$ be the number of qudits, and write $E = \F_p^q$.
	Write $\mathcal A^{(q)} := \F_p[x_1,\ldots,x_q]_{\leq (p-1,\ldots,p-1)} = \spn\{ x_1^{i_1} \cdots x_q^{i_q} \mid i_1,\ldots,i_q \in \{0,1,\ldots,p-1\} \}$ for the $\F_p$-vector space of polynomials in $q$ variables $x_1,\ldots,x_q$ whose degree in each variable $x_i$ is at most $p - 1$.
	For $d \geq 0$, write $\mathcal P_d^{(q)}$ for the space of all polynomials of total degree $\leq d$, and write $\mathcal A_d^{(q)} := \mathcal A^{(q)} \cap \mathcal P_d^{(q)}$.
	
	For $a \in E$, the translation operator $T_a$ and finite difference operator $\Delta_a$ are defined as in \cref{def:translation}.
	The group $G_{p,d}^{(q)}$ is defined as the semidirect product $G_{p,d}^{(q)} = \mathcal A_d^{(q)} \rtimes E$, where $E$ acts on $\mathcal A_d^{(q)}$ by translation.
\end{definition}

The basic properties of $G_{p,d}^{(q)}$ are as in \cref{sec:construction}.
In particular, this group has a faithful unitary representation $\rho : G_{p,d}^{(q)} \to U(\mathcal H)$ acting as Pauli-like matrices on the Hilbert space $\mathcal H = \C^E \cong (\C^p)^{\otimes q}$, given by $\rho(f, a)  = D_f S_a$ where
\[  D_f |z\rangle = \omega^{f(z)} |z\rangle , \qquad S_a |z\rangle = |z-a\rangle. \]

Since $G_{p,p-1}^{(q)}$ contains $G_{p,p-1}^{(2)}$ as a subgroup whenever $q \geq 2$, it follows from \cref{thm:main} that there exist linear systems over $\Z/p\Z$ with no classical solution which do admit a group-valued solution taking values in $G_{p,p-1}^{(q)}$ for all $q \geq 2$.
The main result of this section shows that this fails for all $q$ if we decrease the degree of the polynomials by one.

\begin{theorem}
	\label{thm:threshold-restated}
	Let $p$ be prime and $q \geq 1$.
	Then every linear system with a group-valued solution in $G_{p,p-2}^{(q)}$ also has a classical solution.
\end{theorem}

This contains \cref{thm:threshold} as a special case ($q = 2$).

To prove \cref{thm:threshold-restated}, we need a number of propositions and lemmas.
The proof will be completed at the end of this section, after \cref{thm:p-2}.
First, we recall the following well-known property of the number of monomials of degree at most $d$ in $q$ variables.

\begin{proposition}
	\label{prop:pol}
	For all $q \geq 1$ and $d \geq 0$, we have $\dim ( \mathcal P_d^{(q)} ) = \binom{d + q}{d}$.\qed
\end{proposition}

\begin{definition}
	For $a \in E$, write $N_a := \{g : E \to \F_p \mid \Delta_a g = 0\}$.
	If $d \leq p - 1$, then the evaluation map $\ev : \F_p[x_1,\ldots,x_q] \to \F_p^E$ is injective on $\mathcal P_d^{(q)}$, so we may identify $\mathcal P_d^{(q)}$ with a subset of $\F_p^E$ and write $N_a \cap \mathcal P_d^{(q)}$ for the set of all polynomials $g \in \mathcal P_d^{(q)}$ such that $\Delta_a g = 0$.
\end{definition}

\begin{proposition}
	\label{prop:kernel}
	If $a \neq 0$, then $N_a = N_{ra}$ for all $r \in \F_p^\times$.
	Moreover, if $1 \leq t \leq q$ and $a_1,\ldots,a_t \in E$ are linearly independent, then for all $0 \leq d \leq p - 1$ we have
	\[ \dim \left( N_{a_1} \cap \cdots \cap N_{a_t} \cap \mathcal P_d^{(q)} \right) = \dim \left( \mathcal P_d^{(q-t)} \right) = \binom{d + q - t}{d}. \]
\end{proposition}
\begin{proof}
	Clearly $\Delta_a g = 0$ if and only if $g$ is constant along each line parallel to $\spn(a)$.
	It follows immediately that $N_a = N_{ra}$ for all $r \in \F_p^\times$.
	
	Now assume that $a_1,\ldots,a_t \in E$ are linearly independent.
	Choose $a_{t+1},\ldots,a_q \in E$ such that $a_1,\ldots,a_q$ forms a basis of $E$.
	Write $F := \F_p^{q - t}$, and note that $N_{a_1} \cap \cdots \cap N_{a_t}$ is isomorphic to $\F_p^F$ via the linear maps $\Phi : N_{a_1} \cap \cdots \cap N_{a_t} \to \F_p^F$ and $\Psi : \F_p^F \to N_{a_1} \cap \cdots \cap N_{a_t}$ given by
	\begin{align*}
		\Phi(g)(y) &= g(y_1 a_{t+1} + \cdots + y_{q-t} a_q), \\
		\Psi(h)(z) &= h(Mz),
	\end{align*}
	where $M \in \F_p^{(q - t) \times q}$ is the matrix encoding the projection $\lambda_1 a_1 + \cdots + \lambda_q a_q \mapsto (\lambda_{t+1} , \ldots , \lambda_q)$.
	Since $\Phi$ and $\Psi$ act by substituting linear combinations of the coordinates, plugging in polynomials and expanding the expression shows that $\deg(\Phi(g)) \leq \deg(g)$ and $\deg(\Psi(h)) \leq h$.
	It follows that $\Phi$ and $\Psi$ restrict to isomorphisms between $N_{a_1} \cap \cdots \cap N_{a_t} \cap \mathcal P_d^{(q)}$ and $\mathcal P_d^{(q-t)}$.
	This proves the first dimension equation; the second follows from \cref{prop:pol}.
\end{proof}

\begin{proposition}
	\label{prop:antiderivative}
	Let $0 \leq d \leq p - 2$.
	Then for every $f \in \mathcal P_d^{(q)}$ and $a \neq 0$, there exists $g \in \mathcal P_{d+1}^{(q)}$ such that $f = \Delta_a g$.
\end{proposition}
\begin{proof}
	The map $\Delta_a$ restricts to a linear map $\Delta_{a,d+1} : \mathcal P_{d+1}^{(q)} \to \mathcal P_d^{(q)}$.
	By \cref{prop:kernel}, we have $\dim(\ker(\Delta_{a,d+1})) = \dim(N_a \cap \mathcal P_{d+1}^{(q)}) = \binom{d + q}{d + 1}$.
	Since $\dim(\mathcal P_{d+1}^{(q)}) = \binom{d + q + 1}{d + 1}$ and $\dim(\mathcal P_d^{(q)}) = \binom{d + q}{d}$, we find that
	\begin{align*}
		\rank(\Delta_{a,p-1}) &= \dim(\mathcal P_{d+1}^{(q)}) - \dim(\ker(\Delta_{a,d+1}))\\
		&= \binom{d + q + 1}{d + 1} - \binom{d + q}{d + 1} = \binom{d + q}{d} = \dim(\mathcal P_d^{(q)}).
	\end{align*}
	This shows that $\Delta_{a,d+1}$ is surjective.
\end{proof}

\begin{corollary}
	\label{cor:order-p}
	Every element of $G_{p,p-2}^{(q)}$ has order $p$.
\end{corollary}
\begin{proof}
	Let $(f,a) \in G_{p,p-2}^{(q)}$.
	If $a = 0$, then evidently $(f,a)^p = (pf,pa) = (0,0)$.
	If $a \neq 0$, choose $g \in \mathcal P_{p-1}^{(q)}$ such that $f = \Delta_a g$.
	One easily shows by induction that $(f,a)^r = (\Delta_{ra} g , ra)$ for all $r \geq 0$.
	Therefore, $(f,a)^p = (\Delta_0 g , 0) = (0,0)$.
\end{proof}

It will be convenient to exclude certain low degree terms from polynomials.
For this we use the following notation, extending \cref{def:polys-q}.

\begin{definition}
	For $0 \leq d_1 \leq d_2$, define
	\begin{align*}
		\mathcal P_{[d_1,d_2]}^{(q)} &:= \spn\{ x_1^{i_1} \cdots x_q^{i_q} \mid d_1 \leq i_1 + \cdots + i_q \leq d_2\} \subseteq \F_p[x_1,\ldots,x_q], \\
		\mathcal A_{[d_1,d_2]}^{(q)} &:= \mathcal P_{[d_1,d_2]}^{(q)} \cap \mathcal A^{(q)} = \spn\{ x_1^{i_1} \cdots x_q^{i_q} \mid i_1,\ldots,i_q \in \{0,1,\ldots,p-1\}, d_1 \leq i_1 + \cdots + i_q \leq d_2\}.
	\end{align*}
\end{definition}

\begin{proposition}
	\label{prop:c}
	Let $(f,a) \in G_{p,p-2}^{(q)}$ with $a \neq 0$.
	Then $f$ can be written as $f = \Delta_a g + c$ for some $c \in F_p$ and some $g \in \mathcal A_{[2,p-1]}^{(q)}$.
	Moreover, if $f = \Delta_a g + c = \Delta_a g' + c'$ are two ways of doing so, then $c = c'$.
\end{proposition}
\begin{proof}
	By \cref{prop:antiderivative}, there exists $h \in \mathcal A_{p-1}^{(q)} = \mathcal P_{p-1}^{(q)}$ such that $f = \Delta_a h$.
	Let $g \in \mathcal A_{[2,p-1]}^{(q)}$ be the polynomial obtained from $h$ by removing the constant and linear terms.
	Then $\deg(h - g) \leq 1$, so $c := \Delta_a (h - g)$ is constant.
	Therefore, $f = \Delta_a g + c$ is the desired expression for $f$.
	
	Now let $f = \Delta_a g + c = \Delta_a g' + c'$ be two such expressions.
	Choose some $b \in E$ such that $b \cdot a = 1$, and define $h := g + c (b_1 x_1 + \cdots + b_q x_q)$ and $h' := g' + c' (b_1 x_1 + \cdots + b_q x_q)$.
	Since $\Delta_a (b_1 x_1 + \cdots + b_q x_q) = b_1 a_1 + \cdots + b_q a_q = 1$, we have $f = \Delta_a h = \Delta_a h'$.
	Therefore, $\Delta_a (h - h') = 0$, so $h - h'$ is constant along every line parallel to $\spn(a)$.
	
	Define $k \in \F_p[t]_{\leq p - 1}$ by $k(t) := h(ta) - h'(ta)$.
	Since $h - h'$ is constant (as a function) along $\spn(a)$, and since $\deg(k) \leq p - 1$, it follows that $k$ is a constant polynomial.
	Noting that $b_1 t a_1 + \cdots + b_q t a_q = t$ and recalling that $g$ and $g'$ have no terms of degree $\leq 1$, we see that the coefficient of $t$ in $k(t)$ is equal to $c - c'$.
	Since $k$ is a constant polynomial, we have $c - c' = 0$.
\end{proof}

\begin{remark}
	The constant $c$ in \cref{prop:c} is not necessarily the constant term of $f$.
	For example, if $f = \Delta_{e_1} x_1^2 + c = (x_1 + 1)^2 - x_1^2 + c = 2x_1 + 1 + c$, then the constant term of $f$ is $c + 1$, not $c$.
\end{remark}

\begin{lemma}
	\label{lem:pow}
	Let $(f,a) \in G_{p,p-2}^{(q)}$ with $a \neq 0$, and write $f = \Delta_a g + c$ as in \cref{prop:c}.
	Then $(f,a)^r = (\Delta_{ra} g + rc , ra)$ for all $r \geq 0$.
\end{lemma}
\begin{proof}
	Choose $b \in E$ such that $b \cdot a = 1$, and define $h := g + c(b_1 x_1 + \cdots + b_q x_q)$.
	Since $\Delta_a (b_1 x_1 + \cdots + b_q x_q) = b_1 a_1 + \cdots + b_q a_q = 1$, we have $f = \Delta_a h$.
	By the argument from \cref{cor:order-p}, it follows that $(f,a)^r = (\Delta_{ra} h , ra) = (\Delta_{ra} g + rc , ra)$.
\end{proof}

\begin{lemma}
	\label{lem:rank-1}
	Let $g \in \mathcal P_{[2,p-1]}^{(q)}$ and $a \neq 0$ such that $\Delta_a^2 g = 0$.
	Then the constant term of $\Delta_a g$ is zero.
\end{lemma}
\begin{proof}
	Define $k \in \F_p[t]$ by $k(t) := g(ta)$.
	Since $\Delta_a^2 g = 0$, we see that $k$ is an affine function.
	Since $\deg(k) \leq \deg(g) \leq p - 1$, it follows that $k$ is also an affine polynomial; that is, $\deg(k) \leq 1$.
	But since we assumed that $g$ has no linear or constant terms, neither does $k$, so in fact we have $k = 0$.
	Finally, since $\deg(g) \leq p - 1$, the constant term of $\Delta_a g$ is equal to $(\Delta_a g)(0) = g(a) - g(0) = k(1) - k(0) = 0$.
\end{proof}

\begin{lemma}
	\label{lem:rank-2}
	If $(f,a) , (g,b) \in G_{p,p-2}^{(q)}$ commute and $a$ and $b$ are linearly independent, then there exists $h \in \mathcal P_{p-1}^{(q)}$ such that $f = \Delta_a h$ and $g = \Delta_b h$.
\end{lemma}
\begin{proof}
	Since $(f + T_a g , a+b) = (f,a)(g,b) = (g,b)(f,a) = (g + T_b f , a + b)$, we have $\Delta_a g = \Delta_b f$.
	Using \cref{prop:antiderivative}, choose $h,k \in \mathcal P_{p-1}^{(q)}$ such that $f = \Delta_a h$ and $g = \Delta_b k$.
	Since $\Delta_a \Delta_b = \Delta_b \Delta_a$ and $\Delta_b \Delta_a h = \Delta_b f = \Delta_a g = \Delta_a \Delta_b k$, we have $\Delta_a \Delta_b (h - k) = 0$.
	
	Let $\Delta_{a,p-1} : \mathcal P_{p-1}^{(q)} \to \mathcal P_{p-2}^{(q)}$ denote the restriction of $\Delta_a$ to $\mathcal P_{p-1}^{(q)}$, and likewise define $\Delta_{a,p-2}$, $\Delta_{b,p-1}$ and $\Delta_{b,p-2}$.
	By \cref{prop:antiderivative}, each of these maps is surjective.
	Define $\Phi : \mathcal P_{p-1}^{(q)} \to \mathcal P_{p-3}^{(q)}$ to be the diagonal of the following commutative diagram:
	\[ \begin{tikzcd}[row sep=large,column sep=huge]
		\mathcal P_{p-1}^{(q)} \arrow[r, "\Delta_{a,p-1}"] \arrow[d, "\Delta_{b,p-1}", swap] \arrow[dr, dashed, "\Phi"] & \mathcal P_{p-2}^{(q)} \arrow[d, "\Delta_{b,p-2}"] \\
		\mathcal P_{p-2}^{(q)} \arrow[r, "\Delta_{a,p-2}"]                                                              & \mathcal P_{p-3}^{(q)}
	\end{tikzcd} \]
	Since the other maps in the diagram are all surjective, so is $\Phi$.
	We claim that $\ker(\Phi) = \ker(\Delta_{a,p-1}) + \ker(\Delta_{b,p-1})$.\hair\footnote{Here the sum of the two subspaces denotes the \emph{Minkowski sum}, defined as $A + B = \{a + b \mid a \in A , b \in B\}$. We note that this is not a direct sum, since $\ker(\Delta_{a,p-1})$ and $\ker(\Delta_{b,p-1})$ have nontrivial intersection.}
	It is clear from the diagram that $\ker(\Delta_{a,p-1})$ and $\ker(\Delta_{b,p-1})$ are both contained in $\ker(\Phi)$, so we have $\ker(\Delta_{a,p-1}) + \ker(\Delta_{b,p-1}) \subseteq \ker(\Phi)$.
	Equality follows by a dimension count, using \cref{prop:pol,prop:kernel} and the fact that $\Phi$ is surjective:
	\begin{align*}
		\dim(\ker(\Phi)) &= \dim(\mathcal P_{p-1}^{(q)}) - \rank(\Phi) \\
		&= \dim(\mathcal P_{p-1}^{(q)}) - \dim(\mathcal P_{p-2}^{(q)}) + \dim(\mathcal P_{p-2}^{(q)}) - \dim(\mathcal P_{p-3}^{(q)}) \\
		&= \binom{p + q - 1}{p - 1} - \binom{p + q - 2}{p - 2} + \binom{p + q - 2}{p - 2} - \binom{p + q - 3}{p - 3} \\
		&= \binom{p + q - 2}{p - 1} + \binom{p + q - 3}{p - 2} \\
		&= \binom{p + q - 2}{p - 1} + \binom{p + q - 2}{p - 1} - \binom{p + q - 3}{p - 1} \\
		&= \dim(N_a \cap \mathcal P_{p-1}^{(q)}) + \dim(N_b \cap \mathcal P_{p-1}^{(q)}) - \dim(N_a \cap N_b \cap \mathcal P_{p-1}^{(q)}) \\
		&= \dim(\ker(\Delta_{a,p-1})) + \dim(\ker(\Delta_{b,p-1})) - \dim(\ker(\Delta_{a,p-1}) \cap \ker(\Delta_{b,p-1})) \\
		&= \dim(\ker(\Delta_{a,p-1}) + \ker(\Delta_{b,p-1})).
	\end{align*}
	To complete the proof, note that since $h - k \in \ker(\Phi) = \ker(\Delta_{a,p-1}) + \ker(\Delta_{b,p-1})$, we may choose $h' \in \ker(\Delta_{a,p-1})$ and $k' \in \ker(\Delta_{b,p-1})$ such that $h - k = h' - k'$.
	Now $h - h' = k - k'$ gives the desired function, since $\Delta_a (h - h') = f - 0 = f$ and $\Delta_b (k - k') = g - 0 = g$.
\end{proof}

We now come to the main result of this section.
We define a function $\lambda : G_{p,p-2}^{(q)} \to \F_p$ that is locally additive and maps $J$ to $1$.
Using this $\lambda$, every operator solution to a linear system taking values in $G_{p,p-2}^{(q)}$ can be mapped to a classical solution.

\begin{definition}
	\label{def:lambda}
	We define the function $\lambda : G_{p,p-2}^{(q)} \to \F_p$ by
	\[ \lambda(f,a) = \begin{cases}
		\text{constant term of $f$}, &\quad \text{if $a = 0$}; \\
		\text{$c$ as in \cref{prop:c}}, &\quad \text{if $a \neq 0$}.
	\end{cases} \]
\end{definition}

\begin{theorem}
	\label{thm:p-2}
	The function $\lambda$ defined in \cref{def:lambda} satisfies
	\[ \lambda(uv) = \lambda(u) + \lambda(v), \qquad \text{for all $u,v \in G_{p,p-2}^{(q)}$ with $uv = vu$}. \]
	Furthermore, $\lambda(J) = 1$.
\end{theorem}
\begin{proof}
	Let $u = (f,a)$ and $v = (g,b)$ be commuting elements of $G_{p,p-2}^{(q)}$.
	Define $W := \spn\{a, b\} \subseteq E$.
	We distinguish three cases.
	
	Suppose first that $\dim(W) = 0$.
	Then $u = (f,0)$, $v = (g,0)$, and clearly $\lambda$ is additive.
	
	Second, suppose that $\dim(W) = 1$.
	Assume without loss of generality that $a \neq 0$ (interchange $u$ and $v$ if necessary), and choose $r \in \F_p$ such that $b = r a$.
	Since $u$ and $v$ commute, the subgroup $\langle u, v \rangle := \{ u^s v^t \mid s,t \in \F_p \} \subseteq G_{p,p-2}^{(q)}$ they generate is abelian.
	The product $u^r v^{-1} = v^{-1} u^r$ can be written as $(k,0)$ for some $k \in \mathcal P_{p-2}^{(q)}$ and also commutes with $u$ and $v$.
	Commuting with $u$ means that $\Delta_a k = 0$, and by \cref{prop:kernel} it follows that $\Delta_{sa} k = 0$ for all $s \in \F_p^\times$.
	Using \cref{prop:c}, choose $h \in \mathcal P_{[2,p-1]}^{(q)}$ and $c \in \F_p$ such that $f = \Delta_a h + c$.
	By \cref{lem:pow}, we have $u^r = (\Delta_{ra} h + rc , ra)$, so we can write
	\begin{align*}
		v &= (-k,0) \cdot u^r = (-k + \Delta_{r a} h + r c, ra) = (r c - k + \Delta_b h , b), \\
		uv &= (-k,0) \cdot u^{r+1} = (-k + \Delta_{(r+1)a} h + (r+1)c , (r+1)a) = (rc + c - k + \Delta_{a+b} h , a + b).
	\end{align*}
	Let $c'$ be the constant term of $k$.
	We show that $\lambda(v) = rc - c'$.
	If $b = 0$, then $\Delta_b h = 0$, so the constant term of $g = rc - k$ is $rc - c'$.
	If $b \neq 0$, apply \cref{prop:c} to find $k' \in \mathcal P_{[2,p-1]}^{(q)}$ and $c'' \in \F_p$ such that $k = \Delta_b k' + c''$.
	By \cref{lem:rank-1}, the constant term of $\Delta_b k'$ is zero, so we have $c'' = c'$.
	Now we may write $v = (rc - c' + \Delta_b (h - k') , b)$ where $h - k' \in \mathcal P_{[2,p-1]}^{(q)}$, so we find $\lambda(v) = rc - c'$.
	This proves our claim.
	An analogous argument shows that $\lambda(uv) = rc + c - c'$.
	Since $\lambda(u) = c$, it follows that $\lambda(u) + \lambda(v) = \lambda(uv)$.
	
	Finally, suppose that $\dim(W) = 2$; that is, $a$ and $b$ are linearly independent.
	It follows from \cref{lem:rank-2} that there exists $h \in \mathcal P_{p-1}^{(q)}$ such that $u = (\Delta_a h , a)$ and $v = (\Delta_b h , b)$.
	A simple computation shows that
	\begin{align*}
		(f,a)(g,b) &= (\Delta_a h , a)(\Delta_b h , b) = (\Delta_a h + T_a \Delta_b h , a + b)\\
		&= (h(z + a) - h(z) + h(z + a + b) - h(z + a) , a + b) = (\Delta_{a+b} h , a+b).
	\end{align*}
	Let $h' \in \mathcal P_{[2,p-1]}^{(q)}$ be the polynomial obtained from $h$ by removing the constant and linear terms.
	Then we have $\lambda(u) = \Delta_a (h - h')$, $\lambda(v) = \Delta_b (h - h')$, and $\lambda(uv) = \Delta_{a+b} (h - h')$
	Writing $h - h' = \alpha_1 x_1 + \cdots + \alpha_q x_q + \beta$, we see that $\Delta_w (h - h') = \alpha_1 w_1 + \cdots + \alpha_q w_q$ for all $w \in E$.
	In particular, $\Delta_w (h - h')$ is linear in $w$.
	This shows that $\lambda(uv) = \lambda(u) + \lambda(v)$.
	
	Since $J = (1,0)$, it is immediate from the definition that $\lambda(J) = 1$.
\end{proof}

\begin{proof}[Proof of \cref{thm:threshold-restated}]
	Applying the function $\lambda$ from \cref{def:lambda} turns any group-valued solution taking values in $G_{p,p-2}^{(q)}$ into a classical solution.
\end{proof}

\begin{remark}
	We point out that not one but several steps of the proofs in this section fail when we increase the degree by one.
	
	First, the dimension formula in \cref{prop:kernel} fails when $d \geq p$, since the evaluation map $\ev : \F_p[x_1,\ldots,x_q] \to \F_p^E$ is no longer injective on $\mathcal P_d^{(q)}$ in this case.
	
	Second, \cref{prop:antiderivative} is no longer true for $d = p - 1$.
	For example, let $f = x_1^{p-1}$ and $a = e_1$.
	Since $(x_1 + 1)^p = x_1^p + 1$ as polynomials over $\F_p$, we have $(x_1 + 1)^p - x_1^p = 1$.
	Therefore, $\deg(g(x_1 + 1) - g(x_1)) \leq p - 2$ for all univariate polynomials $g \in \F_p[x_1]$ with $\deg(g) \leq p$, and more generally, $\deg_{x_1}(g(x_1 + 1, x_2, \ldots, x_q) - g(x_1, \ldots, x_q)) \leq p - 2$ for all $g \in \mathcal P_p^{(q)}$.
	This shows that $\Delta_{e_1} g \neq f$ for all $g \in \mathcal P_p^{(q)}$.

	Third, contrary to $G_{p,p-2}^{(q)}$ (see \cref{cor:order-p}), not every element in $G_{p,p-1}^{(q)}$ has order $p$.
	For example, a direct computation shows that $(x_1^{p-1},e_1)^r = (\sum_{i=0}^{r-1} (x_1 + i)^{p-1} , r e_1)$ for all $r \geq 0$, so $(x_1^{p-1},e_1)^p = (-1,0) \neq (0,0)$.
	
	Finally, \Cref{prop:c} does not hold in $G_{p,d}^{(q)}$ when $d \geq p - 1$.
	For existence of $g$ and $c$, we now need to add the extra condition that $(f,a)$ has order $p$ in $G_{p,d}^{(q)}$.
	More importantly, uniqueness of $c$ can also fail, for now it can happen that $k(t) = c_1 + c_2(t - t^p)$ as a polynomial.
	This still defines a constant function $\F_p \to \F_p$, but as a polynomial the coefficient of $t$ is not necessarily $0$.
	This happens because the identification $\F_p^E \cong \F_p[x,y] / (x_1^p - x_1, \ldots, x_q^p - x_q)$ is not basis-independent.
	We leave the details of this explanation to the interested reader.
\end{remark}

\small

\bibliographystyle{my-alphaurl}
\bibliography{contextuality}

@preamble{ "\newcommand{\DobbendeBruyn}{van Dobben de Bruyn} " }

@preamble{ "\newcommand{\Do}{Dob} " }

@preamble{ "\providecommand{\Zivny}{Živný} " }

@preamble{ "\providecommand{\Zi}{Ziv} " }

@article{Aravind-Bell,
	author={Aravind, P.K.},
	title={Bell's Theorem Without Inequalities and Only Two Distant Observers},
	journal={Found. Phys. Lett.},
	volume={15},
	number={4},
	pages={397--405},
	year={2002},
	doi={10.1023/A:1021272729475},
}

@misc{Arkhipov,
	author={Arkhipov, Alex},
	title={Extending and Characterizing Quantum Magic Games},
	year={2012},
	note={Preprint},
	url={https://arxiv.org/abs/1209.3819},
}

@article{Brassard-Broadbent-Tapp,
	author={Brassard, Gilles and Broadbent, Anne and Tapp, Alain},
	title={Quantum pseudo-telepathy},
	journal={Found. Phys.},
	volume={35},
	number={11},
	pages={1877--1907},
	year={2005},
	doi={10.1007/s10701-005-7353-4},
}

@article{Budroni-Cabello-Guhne-Kleinmann-Larsson,
	author={Budroni, Costantino and Cabello, Adán and Gühne, Otfried and Kleinmann, Matthias and Larsson, Jan-{\AA}ke},
	title={{Kochen}-{Specker} contextuality},
	journal={Rev. Mod. Phys.},
	volume={94},
	lpages={045007 (62~p.)},
	year={2022},
	doi={10.1103/RevModPhys.94.045007},
}

@inproceedings{Bulatov-Zivny,
	author={Bulatov, Andrei A. and \Zivny{}, Stanislav},
	title={Satisfiability of Commutative vs. Non-Commutative {CSPs}},
	booktitle={52nd International Colloquium on Automata, Languages, and Programming (ICALP 2025)},
	year={2025},
	pages={37 (18~p.)},
	doi={10.4230/LIPIcs.ICALP.2025.37},
}

@misc{Ciardo,
	author={Ciardo, Lorenzo},
	title={Linear equations mod \(n\) are pseudo-telepathic},
	note={Preprint},
	year={2026},
	url={https://arxiv.org/abs/2609.14632},
}

@article{Chung-Okay-Sikora,
	author={Chung, Ho Yiu and Okay, Cihan and Sikora, Igor},
	title={Simplicial techniques for operator solutions of linear constraint systems},
	journal={Topology Appl.},
	volume={348},
	year={2024},
	pages={108883 (39~p.)},
	doi={10.1016/j.topol.2024.108883},
}

@article{Cleve-Liu-Slofstra,
	author={Cleve, Richard and Liu, Li and Slofstra, William},
	title={Perfect commuting-operator strategies for linear system games},
	journal={J. Math. Phys.},
	volume={58},
	number={1},
	year={2017},
	pages={012202 (7~p.)},
	doi={10.1063/1.4973422},
}

@incollection{Cleve-Mittal,
	author={Cleve, Richard and Mittal, Rajat},
	title={Characterization of binary constraint system games},
	booktitle={Automata, Languages, and Programming (ICALP 2014)},
	pages={320--331},
	publisher={Springer},
	address={Berlin},
	__series={Lecture Notes in Comput. Sci.},
	__volume={8572},
	year={2014},
	doi={10.1007/978-3-662-43948-7_27},
}

@misc{Coladangelo-Stark,
	author={Coladangelo, Andrea and Stark, Jalex},
	title={Robust self-testing for linear constraint system games},
	note={Preprint},
	year={2019},
	url={https://arxiv.org/abs/1709.09267},
}

@misc{Dobben-Roberson-Negami,
	author={\DobbendeBruyn{}, Josse and Roberson, David E.},
	title={$1$-$2$-$\infty$ and beyond: Connections between solution groups, harmonic homomorphisms, and {N}egami's conjecture},
	note={In preparation},
	year={2026},
}

@misc{Frembs,
	author={Frembs, Markus},
	title={A complete classification of the existence of finite-dimensional quantum solutions to inconsistent linear constraint systems},
	note={Preprint},
	year={2026},
	url={https://arxiv.org/abs/2609.16490},
}

@article{Frembs-Okay-Chung,
	author={Frembs, Markus and Okay, Cihan and Chung, Ho Yiu},
	title={No quantum solutions to linear constraint systems in odd dimension from {P}auli group and diagonal {C}liffords},
	journal={Quantum},
	volume={9},
	pages={1583 (23~p.)},
	year={2025},
	doi={10.22331/q-2025-01-08-1583},
}

@misc{Frembs-Okay-Chung-2,
	author={Frembs, Markus and Okay, Cihan and Chung, Ho Yiu},
	title={No quantum solutions to linear constraint systems from monomial measurement-based quantum computation in odd prime dimension},
	note={Preprint},
	year={2026},
	url={https://arxiv.org/abs/2601.11367},
}

@article{Goldberg,
	author={Goldberg, Adina},
	title={Synchronous linear constraint system games},
	journal={J. Math. Phys.},
	volume={62},
	number={3},
	pages={032201 (9~p.)},
	year={2021},
	doi={10.1063/5.0025647},
}

@article{Kim-Paulsen-Schafhauser,
	author={Kim, Se-Jin and Paulsen, Vern and Schafhauser, Christopher},
	title={A synchronous game for binary constraint systems},
	journal={J. Math. Phys.},
	volume={59},
	number={3},
	pages={032201 (17~p.)},
	year={2018},
	doi={10.1063/1.4996867},
}

@article{Mermin,
	author={Mermin, N. David},
	title={Simple unified form for the major no-hidden-variables theorems},
	journal={Phys. Rev. Lett.},
	volume={65},
	number={27},
	pages={3373--3376},
	year={1990},
	doi={10.1103/PhysRevLett.65.3373},
}

@article{Mermin-pentagram,
	author={Mermin, N. David},
	title={Hidden variables and the two theorems of {J}ohn {B}ell},
	journal={Rev. Mod. Phys.},
	volume={65},
	number={3},
	pages={803--815},
	year={1993},
	doi={10.1103/RevModPhys.65.803},
}

@article{PRSS,
	author={Paddock, Connor and Russo, Vincent and Silverthorne, Turner and Slofstra, William},
	title={{A}rkhipov's theorem, graph minors, and linear system nonlocal games},
	journal={Algebr. Comb.},
	volume={6},
	number={4},
	pages={1119--1162},
	year={2023},
	doi={10.5802/alco.292},
}

@article{Peres,
	author={Peres, Asher},
	title={Incompatible results of quantum measurements},
	journal={Phys. Lett. A},
	volume={151},
	number={3,4},
	year={1990},
	pages={107--108},
	doi={10.1016/0375-9601(90)90172-K},
}

@article{Qassim-Wallman,
	author={Qassim, Hammam and Wallman, Joel J.},
	title={Classical vs quantum satisfiability in linear constraint systems modulo an integer},
	journal={J. Phys. A, Math. Theor.},
	volume={53},
	number={38},
	pages={385304 (16~p.)},
	year={2020},
	doi={10.1088/1751-8121/aba306},
}

@article{Slofstra-Tsirelson,
	author={Slofstra, William},
	title={{T}sirelson's problem and an embedding theorem for groups arising from non-local games},
	journal={J. Am. Math. Soc},
	volume={33},
	number={1},
	year={2020},
	pages={1--56},
	doi={10.1090/jams/929},
}

@article{Slofstra-closed,
	author={Slofstra, William},
	title={The set of quantum correlations is not closed},
	journal={Forum Math. Pi},
	volume={7},
	pages={e1 (41~p.)},
	year={2019},
}

@misc{Slofstra-Zhang,
	author={Slofstra, William and Zhang, Lu-Ming},
	title={Operator solutions of linear systems and small cancellation},
	year={2024},
	note={Preprint},
	url={https://arxiv.org/abs/2412.10305},
}

\end{document}